\documentclass[10pt,journal]{IEEEtran}

\usepackage[T1]{fontenc}
\usepackage[utf8]{inputenc}
\usepackage{lmodern}
\usepackage{amsmath,amssymb,mathtools,amsthm}
\usepackage{graphicx}
\usepackage{booktabs,array}
\usepackage{cite}
\usepackage{url}

\theoremstyle{plain}
\newtheorem{theorem}{Theorem}
\newtheorem{proposition}[theorem]{Proposition}
\newtheorem{lemma}[theorem]{Lemma}
\newcommand{\E}{\mathbb{E}}
\newcommand{\Pp}{\mathbb{P}}
\newcommand{\1}[1]{\mathbf{1}\!\left\{#1\right\}}
\newcommand{\Exp}{\operatorname{Exp}}
\DeclareMathOperator{\Var}{Var}

\title{Collision-Free Schedule Acquisition through Slotted ALOHA: Scaling Laws and Fluctuations}
\author{Pl\'inio~S.~Dester%
\thanks{P. S. Dester is with the School of Engineering, S\~ao Paulo State University (UNESP), S\~ao Jo\~ao da Boa Vista 13876-750, Brazil (e-mail: plinio.dester@unesp.br). He is also with the University of Campinas (UNICAMP), School of Electrical and Computer Engineering, Campinas, 13083852, Brazil. This work was supported by RNP/MCTI Brasil 6G (01245.020548/2021-07).}}
\date{}

\begin{document}
\maketitle

\begin{abstract}
Decentralized networks can learn collision-free periodic transmission schedules, but acquisition can consume much of a finite operating horizon. We study a reservation protocol for a fixed population of $n$ nodes and a period of $\Gamma_n\ge n$ slots. Initially, each node attempts independently in each slot with probability proportional to $1/n$. A success reserves a slot that the node reuses every $\Gamma_n$ slots, retaining it after collisions. Nodes use only their own acknowledgments and small local memory. We derive an exact finite-population Markov chain and characterize the time until every node has a distinct reservation. As $n\to\infty$, mean acquisition takes $O(n^2)$ slots when $\Gamma_n=n$ and $O(n\log n)$ slots when $\Gamma_n\sim(1+\theta)n$, for fixed $\theta>0$. With a fixed number $\Gamma_n-n$ of spare slots, the standard deviation of acquisition time has the same asymptotic order as its mean. If that number diverges while $\Gamma_n/n$ remains bounded, the ratio of standard deviation to mean tends to zero. A diverging sublinear number of spare slots yields $o(n^2)$ mean acquisition while long-run utilization, the fraction $n/\Gamma_n$ of slots with successful transmissions, tends to one. Simulations validate these laws and compare acquisition and delivered packets per slot with existing access protocols.

\end{abstract}


\section{Introduction}

Distributed medium access control (MAC) can learn reusable collision-free transmission schedules without central slot assignment. Memory-based access, semi-random backoff, and learning protocols use previous transmission outcomes to achieve this \cite{Park2010,He2013,Fang2013}. However, convergence can consume much of a finite operating horizon. We call the time until nodes establish a collision-free periodic schedule the \emph{schedule acquisition time}. We study how a simple decentralized rule accelerates acquisition while maintaining high utilization.

Our starting point is 1-persistent age-threshold slotted ALOHA (1-pTSA) \cite{Dester1Persistent2026,DesterDeterministic2026}. In a network of $n$ nodes, a successful transmitter becomes \emph{scheduled} and attempts again after $\Gamma_n\ge n$ slots. Its \emph{phase} is the slot position repeated every $\Gamma_n$ slots. If a scheduled attempt collides, the node becomes \emph{contending} and attempts independently with probability $p_n\in(0,1)$ per slot until its next success. This counter description is equivalent to the original age-threshold formulation when a fresh update is generated at each attempt. A contending transmission can therefore destroy an existing reservation.

In \emph{role-protected threshold slotted ALOHA} (RP-TSA), a scheduled node retains its phase after collisions. Starting with all nodes contending, each successful contending transmission establishes a distinct reservation without disturbing prior reservations. Nodes share slot timing and preconfigured parameters but have no assigned phases, exchanged neighbor states, or reservation map. Each uses its own acknowledgment (ACK) outcomes, a role flag, and a countdown; the receiver acknowledges packets without assigning slots. Scheduled nodes can sleep between attempts while maintaining slot timing, supporting energy-efficient operation with a simple implementation.

Reservation persistence also appears in self-organizing time-division multiple access and collision-free carrier-sense protocols \cite{Derakhshani2019,Barcelo2011}. When packets are always available, learning MAC (L-MAC) with learning parameter $\beta=1$ uses uniform phase probabilities until success, then retains the successful phase permanently \cite{Fang2013}. It chooses one attempt per period using only a role and phase. RP-TSA instead permits independent attempts in each slot while contending, enabling the exact acquisition analysis below.

The acquired schedule has utilization $n/\Gamma_n$, the fraction of slots delivering a packet. Its \emph{slack} $s_n=\Gamma_n-n$ counts spare slot positions; $\Gamma_n=n$ has zero slack. We characterize the tradeoff between faster acquisition and lower utilization as slack increases. Our contributions are:

\begin{itemize}
\item An exact Markov reduction at reservation-period boundaries, with a transition law independent of the arrangement of reserved phases and deterministic bounds connecting the reduced dynamics to elapsed slots.
\item Under $np_n\to c\in(0,\infty)$, a quadratic acquisition law with persistent relative fluctuations at fixed slack. Growing slack yields concentration: linear slack gives the $n\log n$ scale, while growing sublinear slack gives $o(n^2)$ acquisition with utilization tending to one.
\item Numerical validation and tuning of $p_n$, comparing acquisition and packet delivery against 1-pTSA, L-MAC, and learning zero-collision (L-ZC), including equal nominal attempt rates.
\end{itemize}

The exact analysis concerns permanent reservations. A finite collision limit additionally enables recovery from initially conflicting reservations, and we establish when it preserves the distributional acquisition limits.

\section{Network Model and Protocol}
\label{sec:model}

\subsection{Distributed collision channel}

A fixed finite population of $n\ge2$ homogeneous nodes shares a slotted uplink to one receiver, with slots indexed by $t=0,1,\ldots$. Every node always has a packet available (saturated traffic). Exactly one transmission in a slot delivers one packet; multiple transmissions collide and deliver none. We assume slot synchronization, reliable addressed ACKs, and a single collision domain, excluding capture, erasures, hidden terminals, and population changes. The integer period $\Gamma_n\ge n$ and access probability $p_n\in(0,1)$ are preconfigured. The subscript $n$ records population dependence for the limit $n\to\infty$.

All compared configurations use $n$ to tune their period; $p_n=c/n$ and the L-ZC collision-retention probability $1/(\Gamma_n-n+2)$ also depend on $n$. Once configured, local updates require no population or reservation exchanges. None uses extra scheduling messages or headers beyond ACK and timing support \cite{Fang2013}.

\subsection{RP-TSA with permanent reservations}

Each node is either contending or scheduled. A contending node $i\in\{1,\ldots,n\}$ attempts independently in each slot with probability $p_n$. Success makes it scheduled, with countdown $C_i=\Gamma_n$ for the next slot. A scheduled node transmits when $C_i=1$ and resets its counter to $\Gamma_n$, with or without an ACK; otherwise the counter decrements. A scheduled transmission in slot $t$ therefore repeats in slot $t+\Gamma_n$, with sleep possible between attempts.

The main analysis starts with all nodes contending. A contending success reserves an \emph{empty phase}, which has no scheduled owner; hence scheduled phases remain distinct. An empty phase may still carry contending transmissions. A phase with exactly one scheduled owner is a \emph{singleton}. A scheduled collision therefore involves contending nodes, and retaining the reservation requires no information about collision multiplicity. Only reservation retention after failure differs from 1-pTSA.

\subsection{Recovery with a finite collision limit}
\label{sec:finite-persistence}

Permanent reservations cannot resolve two scheduled nodes initially sharing a phase: both would repeat their collision forever. An ACK-only recovery variant releases a reservation after $\kappa$ consecutive failed scheduled attempts, where $\kappa\ge1$ is an integer. A local failure counter starts at zero when a reservation is acquired, resets to zero after an ACK, and increments after a failed scheduled attempt. On the $\kappa$th failure the node becomes contending; otherwise it retains its phase. This convention makes $\kappa=1$ equivalent to 1-pTSA and $\kappa=\infty$ equivalent to RP-TSA.

Finite persistence uses only ACKs but can release a singleton hit repeatedly by contending nodes. The contending count can increase, so the exact reduction in Section~\ref{sec:reduction} requires permanent reservations. Proposition~\ref{prop:global} establishes recovery for finite $\kappa$, and Section~\ref{sec:finite-coupling} relates sufficiently growing collision limits to the permanent-reservation model.

\subsection{Acquisition time and delivered service}

Unless stated otherwise, we use \emph{ordinary startup}: all nodes begin contending at time zero and attempt independently in each slot with probability $p_n$. In \emph{forced startup}, every node transmits in slot zero and collides, then uses ordinary access from slot one. No node acquires a reservation, so acquisition takes one additional slot.

For $t\ge0$, let $X_n(t)$ record the set of contending nodes, assignment of scheduled nodes to phases, and current phase after $t$ completed slots. Under ordinary access, this configuration is a finite-state time-homogeneous Markov chain \cite{Norris1997}. A reservation configuration is \emph{binary} if every phase is empty or singleton. Let $\mathcal D_{n,\Gamma_n}$ contain the binary configurations with no contending nodes. Define
\begin{equation}
 \tau_n=\inf\{t\ge0:X_n(t)\in\mathcal D_{n,\Gamma_n}\}.
 \label{eq:tau}
\end{equation}
Acquisition is an observer-defined hitting time that nodes cannot detect locally. The finite-persistence variant additionally records the failure counters; its acquisition time is denoted $\tau_n^{(\kappa)}$.

Let $S_n(t)$ count the successfully delivered packets in the first $t$ slots. Goodput is measured in packets per modeled slot. Its finite-horizon value and normalization by the collision-free utilization are
\begin{equation}
 G_n(t)=\frac{S_n(t)}{t},\qquad
 \widetilde G_n(t)=\frac{\Gamma_n S_n(t)}{nt},\qquad t\ge1.
 \label{eq:goodput}
\end{equation}
The normalization is a steady-service reference, not a finite-horizon upper bound. Once acquisition occurs, periodic service gives $\widetilde G_n(t)\to1$ almost surely as $t\to\infty$.

A \emph{frame} comprises $\Gamma_n$ consecutive slots; for analysis, frame $k\ge0$ spans slots $k\Gamma_n$ through $(k+1)\Gamma_n-1$. RP-TSA needs no shared frame origin or exchanged phase labels. Table~\ref{tab:notation} lists the principal notation, and Table~\ref{tab:requirements} compares protocol requirements. Plain Roman capitals denote random objects. Population subscripts are suppressed on local counters and state counts when unambiguous. The arrows $\xrightarrow{\mathrm d}$ and $\xrightarrow{\mathrm p}$ denote convergence in distribution and probability. All logarithms are natural.

\begin{table*}[t]
\centering
\caption{Principal notation.}
\label{tab:notation}
\small
\renewcommand{\arraystretch}{1.23}
\setlength{\tabcolsep}{5pt}
\begin{tabular}{@{}c >{\raggedright\arraybackslash}p{0.35\textwidth} c >{\raggedright\arraybackslash}p{0.35\textwidth}@{}}
\toprule
Symbol & Meaning & Symbol & Meaning \\
\midrule
$n,\Gamma_n$ & Population and reservation period & $s_n$ & Slack $\Gamma_n-n$ \\
$\theta$ & Linear-slack ratio, $s_n/n\to\theta>0$ & $u_n$ & Collision-free utilization $n/\Gamma_n$ \\
$p_n,c$ & Contending access probability; limit of $np_n$ & $\kappa$ & Optional collision limit \\
$C_i$ & Scheduled countdown & $Q_t,H_t$ & Contending nodes; empty phases before $t$ \\
$\tau_n$ & Completed slots until acquisition & $M_k$ & Contending count at frame boundary $k$ \\
$\boldsymbol{\Pi}_n$ & One-opportunity transition matrix & $g_{n,m}$ & Success probability at an opportunity \\
$G_{n,m}$ & Geometric opportunity holding time & $K_n,\mathcal W_n$ & Absorption frame; weighted opportunities \\
$T_s$ & Fixed-slack limiting acquisition time & $\beta_n^\star$ & Diverging-slack acquisition scale \\
$S_n(t),G_n(t)$ & Delivered packets; finite-horizon goodput & $\widetilde G_n(t)$ & Goodput divided by $n/\Gamma_n$ \\
\bottomrule
\end{tabular}
\end{table*}

\begin{table*}[t]
\centering
\caption{Requirements of decentralized periodic schedulers under saturated traffic.}
\label{tab:requirements}
\small
\renewcommand{\arraystretch}{1.2}
\setlength{\tabcolsep}{4pt}
\begin{tabular}{@{}>{\raggedright\arraybackslash}p{0.15\textwidth} >{\raggedright\arraybackslash}p{0.22\textwidth} >{\raggedright\arraybackslash}p{0.28\textwidth} >{\raggedright\arraybackslash}p{0.29\textwidth}@{}}
\toprule
Protocol & Decision information & Local reservation state & Listening between attempts \\
\midrule
1-pTSA \cite{Dester1Persistent2026} & Own ACK outcome & Role and countdown & None, with slot timing maintained \\
RP-TSA & Own ACK outcome & Role and countdown; failure counter if $\kappa<\infty$ & None, with slot timing maintained \\
L-MAC \cite{Fang2013} & Own ACK outcome & Phase-probability vector and chosen phase & Carrier sensing for wireless local-area network slot timing \\
L-ZC \cite{Fang2013} & Own ACK outcome and idle-phase locations & Chosen phase and preceding frame's idle phases & Idle/busy observations over the frame \\
\bottomrule
\end{tabular}
\par\vspace{4pt}
\begin{minipage}{0.98\textwidth}
\footnotesize
All schemes use a configured period without central phase assignment or neighbor-state exchange. The fixed-slot learner adaptations used here remove carrier sensing for timing; L-ZC still observes idle phases, i.e., phases with no transmission in the preceding frame. L-MAC at $\beta=1$ requires only a role and phase.
\end{minipage}
\end{table*}

\section{Exact Finite-Population Reduction}
\label{sec:reduction}

\begin{proposition}[Collision-free invariance]
\label{prop:invariance}
Every configuration in $\mathcal D_{n,\Gamma_n}$ remains collision-free under either permanent or finite persistence. Each node delivers one packet every $\Gamma_n$ slots, and aggregate utilization is
\begin{equation}
 u_n=\frac{n}{\Gamma_n}.
 \label{eq:cf-metrics}
\end{equation}
\end{proposition}

\begin{IEEEproof}
With no contending nodes and at most one scheduled node per phase, every scheduled transmission succeeds and retains its phase. Empty phases remain empty, and no collision counter can trigger a release.
\end{IEEEproof}

\begin{proposition}[Recovery with finite persistence]
\label{prop:global}
For fixed $n$, $\Gamma_n\ge n$, $0<p_n<1$, and finite $\kappa\ge1$, the collision-limit variant reaches $\mathcal D_{n,\Gamma_n}$ almost surely with finite mean from every admissible reservation configuration and failure-counter state.
\end{proposition}

\begin{IEEEproof}
An admissible state assigns each node to the contending set or one scheduled phase, with failure counters in $\{0,\ldots,\kappa-1\}$. Acquiring a reservation cannot create a phase with multiple scheduled nodes. Every initially shared phase becomes binary within $\kappa\Gamma_n$ slots: while it remains shared, every due attempt fails and increments its owners' counters. Thereafter, prescribe no contending attempts in occupied phases and one designated contending transmitter in each of $m$ distinct empty phases, where $m$ is the current contending population. Such phases exist because $\Gamma_n\ge n$. This event preserves all reservations and gives every contending node a reservation in one frame, with probability at least $\delta_n=p_n^n(1-p_n)^{n\Gamma_n}>0$, uniformly over binary states. The remaining frame count is dominated by a geometric variable with parameter $\delta_n$.
\end{IEEEproof}

\begin{lemma}[Binary invariant and unreserved-slot opportunities]
\label{lem:holes}
With permanent reservations and all nodes initially contending, every scheduled phase has multiplicity at most one. If $Q_t=m$ is the contending population before slot $t$, the number of empty phases is
\begin{equation}
 H_t=s_n+m. \label{eq:holes}
\end{equation}
An occupied phase leaves $m$ unchanged. An \emph{unreserved-slot opportunity}, or \emph{opportunity} for short, is a slot whose phase is empty; there $m$ decreases by one with probability
\begin{equation}
 g_{n,m}=mp_n(1-p_n)^{m-1}. \label{eq:g}
\end{equation}
\end{lemma}

\begin{IEEEproof}
Initially all phases are empty. A unique contending attempt in an empty phase creates a scheduled singleton and removes one contending node from contention. Existing singletons are never erased. Hence $H_t-Q_t=s_n$ is invariant, and the probability of exactly one contending attempt is \eqref{eq:g}.
\end{IEEEproof}

The exact reduction below uses permanent reservations and ordinary startup at time zero. At opportunities, the contending count is a \emph{pure-death Markov chain}: it stays fixed or decreases by one until reaching zero. These analytical observation times require no local detection. In slot time, $Q_t$ alone is generally not Markov, because its next transition also depends on whether the current phase is occupied. Sampling at frame boundaries gives another exact reduction.

\subsection{Arrangement-independent frame kernel}

Set $M_k=Q_{k\Gamma_n}$ for $k\ge0$, so $M_0=n$. Let $\boldsymbol{\Pi}_n$ be the one-opportunity transition matrix, with rows and columns indexed by $0,\ldots,n$, and entries
\begin{equation}
 \begin{aligned}
 [\boldsymbol{\Pi}_n]_{m,m-1}&=g_{n,m}, &&1\le m\le n,\\
 [\boldsymbol{\Pi}_n]_{m,m}&=1-g_{n,m}, &&1\le m\le n,\\
 [\boldsymbol{\Pi}_n]_{0,0}&=1,
 \end{aligned}
 \label{eq:Pn}
\end{equation}
with all other entries zero.

\begin{proposition}[Exact frame chain]
\label{prop:frame}
The frame-boundary contending count is a Markov chain with
\begin{equation}
 \Pp(M_{k+1}=j\mid M_k=m)=\bigl[\boldsymbol{\Pi}_n^{\,s_n+m}\bigr]_{m,j}.
 \label{eq:frame-kernel}
\end{equation}
The right-hand side is the $(m,j)$ entry of the $(s_n+m)$th matrix power of $\boldsymbol{\Pi}_n$, representing $s_n+m$ successive opportunity transitions.
\end{proposition}

\begin{IEEEproof}
A frame beginning with $m$ contending nodes has exactly $s_n+m$ empty phase locations. A newly acquired reservation cannot generate a scheduled transmission within the same frame. Enumerating these locations therefore gives $s_n+m$ consecutive applications of \eqref{eq:Pn}, even though their success probabilities change after each population decrease. Their positions do not enter the transition law.
\end{IEEEproof}

Neither uniform phase locations nor a mixing approximation is used. This kernel also yields a deterministic finite-$n$ check. Write $\pi_{mj}=[\boldsymbol{\Pi}_n^{\,s_n+m}]_{m,j}$ and let $e_m$ be the expected additional number of frames until absorption. Then
\begin{equation}
 e_0=0,\qquad e_m=\frac{1+\sum_{j=0}^{m-1}\pi_{mj}e_j}{1-\pi_{mm}}.
 \label{eq:frame-mean}
\end{equation}
Consequently,
\begin{equation}
 \Gamma_n(e_n-1)<\E\tau_n\le\Gamma_ne_n.
 \label{eq:finite-mean-bracket}
\end{equation}
The frame matrix $\boldsymbol\pi_n=(\pi_{mj})$ also gives exact acquisition probabilities at frame boundaries:
\[
 \Pp\{\tau_n\le k\Gamma_n\}=[\boldsymbol\pi_n^{\,k}]_{n,0},\qquad k\ge0.
\]
At intermediate times, adjacent frame-boundary values bound the probability by monotonicity.

\subsection{From opportunities to elapsed slots}

On each sample path, let $G_{n,m}$ count the empty-phase opportunities visited while $m$ nodes remain contending, including the successful opportunity. Fresh independent access choices imply independent geometric holding times:
\begin{equation}
 G_{n,m}\stackrel{\rm ind}{\sim}\operatorname{Geom}(g_{n,m}),
 \qquad \mathcal W_n=\sum_{m=1}^n\frac{G_{n,m}}{s_n+m}.
 \label{eq:weighted-clock}
\end{equation}
The geometric distribution is supported on $\{1,2,\ldots\}$. The dimensionless work $\mathcal W_n$ counts opportunities weighted by the number available at their current level.

Let $K_n=\inf\{k\ge1:M_k=0\}$. For $s_n>0$,
\begin{equation}
 K_n-1\le\mathcal W_n\le K_n+\frac{n}{s_n}.
 \label{eq:clock-bound}
\end{equation}
Indeed, a complete frame starting at level $m$ visits $s_n+m$ empty locations. If it contains $r$ new reservations, its work is at least one and at most $(s_n+m)/(s_n+m-r)\le1+r/s_n$. The final partial frame obeys the same upper bound. Since at most $n$ reservations are acquired, the total excess term is at most $n/s_n$. For $s_n=0$, the lower bound still holds because each complete frame has $m$ opportunities with weights at least $1/m$.

Combining this observation with $\Gamma_n(K_n-1)<\tau_n\le\Gamma_nK_n$ gives
\begin{align}
 \tau_n&\le\Gamma_n(\mathcal W_n+1), &&s_n\ge0,
 \label{eq:one-sided-clock}\\
 \left|\frac{\tau_n}{\Gamma_n}-\mathcal W_n\right|
 &\le1+\frac{n}{s_n}, &&s_n>0.
 \label{eq:two-sided-clock}
\end{align}
These bounds are the connection between the exact opportunity chain and the slot-time laws below. The one-sided bound controls the entrance from many contending nodes; the two-sided bound controls diverging slack.

\section{Slack-Dependent Acquisition Laws}
\label{sec:analysis}

\subsection{Fixed slack: a nondegenerate random limit}

Write $h_k^{(r)}=\sum_{j=1}^k j^{-r}$, $h_k=h_k^{(1)}$, and $h_0^{(r)}=0$. Let $E_m$, $m\ge1$, be independent unit-mean exponential variables. For an integer $s\ge0$, define
\begin{equation}
 T_s=\sum_{m=1}^{\infty}\frac{E_m}{cm(m+s)},
 \qquad E_m\stackrel{\mathrm{iid}}{\sim}\Exp(1).
 \label{eq:Ts}
\end{equation}
The series converges almost surely and in every fixed-order $L^k$.

\begin{theorem}[Fixed-slack acquisition]
\label{thm:fixed-slack}
Suppose $s_n\to s\in\{0,1,\ldots\}$ and $np_n\to c\in(0,\infty)$. Starting with all nodes contending under ordinary access,
\begin{equation}
 \begin{aligned}
 \frac{\tau_n}{n^2}&\xrightarrow{\mathrm d}T_s,
 &\E\frac{\tau_n}{n^2}&\to\E T_s,\\
 \Var\!\left(\frac{\tau_n}{n^2}\right)&\to\Var(T_s).
 \end{aligned}
 \label{eq:fixed-limit}
\end{equation}
At $s=0$,
\begin{equation}
 \begin{aligned}
 \E T_0&=\frac{\pi^2}{6c},
 &\Var(T_0)&=\frac{\pi^4}{90c^2},\\
 \E e^{-uT_0}&=\frac{\pi\sqrt{u/c}}{\sinh(\pi\sqrt{u/c})}, &&u\ge0,
 \end{aligned}
 \label{eq:s0-moments}
\end{equation}
with the transform interpreted continuously at $u=0$. For integer $s\ge1$ and $u\ge0$, write $a_s(u),b_s(u)=(s\pm\sqrt{s^2-4u/c})/2$ and let $\Gamma_{\!\mathrm E}$ denote the Euler gamma function \cite{NISTDLMF2026}. Then
\begin{align}
 \E T_s&=\frac{h_s}{cs}, \label{eq:fixed-mean}\\
 \Var(T_s)&=\frac{\pi^2/3-h_s^{(2)}-2h_s/s}{c^2s^2}, \label{eq:fixed-var}\\
 \E e^{-uT_s}&=\frac{\Gamma_{\!\mathrm E}(1+a_s(u))\Gamma_{\!\mathrm E}(1+b_s(u))}{\Gamma_{\!\mathrm E}(1+s)}, \label{eq:fixed-lt}
\end{align}
When $4u/c>s^2$, these roots are $x\pm\mathrm{i}y$, with $x=s/2$ and $y=\sqrt{4u/c-s^2}/2$. The gamma functions have conjugate arguments $1+x\pm\mathrm{i}y$, and their product is $|\Gamma_{\!\mathrm E}(1+x+\mathrm{i}y)|^2$. The Laplace transform therefore remains strictly positive and real, independently of the square-root branch.
\end{theorem}

At a fixed contending population $m$, there are $m+s$ empty phases per frame and $ng_{n,m}\to cm$. The slot-time holding period divided by $n^2$ therefore approaches an exponential variable with rate $cm(m+s)$. Appendix~\ref{app:fixed} makes this argument uniform over the entrance from $n$ contending nodes and proves moment convergence. In particular,
\begin{equation}
 \frac{\Var(\tau_n)}{(\E\tau_n)^2}\longrightarrow\frac25
 \qquad(s_n=0).
 \label{eq:relative-variance}
\end{equation}
Thus the coefficient of variation approaches $\sqrt{2/5}\simeq0.6325$: acquisition at zero slack remains random on its mean scale, regardless of the fixed contention intensity $c$.

\subsection{Diverging slack: concentration and crossover}

Define
\begin{equation}
 \beta_n^\star=\frac{\Gamma_n}{p_n}\sum_{m=1}^n\frac1{m(m+s_n)}.
 \label{eq:Bn}
\end{equation}
For $s_n>0$, the finite sum is exactly
\begin{equation}
 \sum_{m=1}^n\frac1{m(m+s_n)}
 =\frac{h_n+h_{s_n}-h_{n+s_n}}{s_n}.
 \label{eq:harmonic}
\end{equation}

\begin{theorem}[Diverging-slack acquisition]
\label{thm:slack}
Suppose $s_n\to\infty$, $\Gamma_n/n$ remains bounded, and $np_n\to c\in(0,\infty)$, with the initialization of Theorem~\ref{thm:fixed-slack}. Then
\begin{equation}
 \frac{\tau_n}{\beta_n^\star}\longrightarrow1\quad\text{in }L^2.
 \label{eq:slack-limit}
\end{equation}
In particular, $\E\tau_n/\beta_n^\star\to1$, and
\begin{align}
 1\ll s_n\ll n:\quad
 &\E\tau_n\sim\frac{n^2\log s_n}{cs_n}, \label{eq:sublinear}\\
 s_n/n\to\theta>0:\quad
 &\frac{\tau_n}{n\log n}\xrightarrow{\mathrm p}\frac{1+\theta}{c\theta},
 \notag\\[-1mm]
 &\frac{\E\tau_n}{n\log n}\to\frac{1+\theta}{c\theta}. \label{eq:linear}
\end{align}
\end{theorem}

The proof in Appendix~\ref{app:slack} combines the independent opportunity holding times with \eqref{eq:two-sided-clock}. Their relative variance and the relative error in converting opportunities to slots both vanish. In particular, $\Gamma_n\sim2n$ gives $\E\tau_n\sim(2/c)n\log n$ and $\tau_n/(n\log n)\xrightarrow{\mathrm p}2/c$.

The mechanism is the supply of empty phases. At zero slack, the final $m$ contending nodes have only $m$ empty phases, giving a removal rate of order $m^2/n^2$. With linear slack, they retain order-$n$ empty phases and a removal rate of order $m/n$. Summing reciprocal rates gives the different acquisition scales. These laws concern permanent reservations initialized with $n$ contending nodes.

There is also an asymptotically vanishing-cost intermediate regime. If $s_n\to\infty$ and $s_n=o(n)$, then \eqref{eq:sublinear} gives $\E\tau_n=o(n^2)$ and \eqref{eq:slack-limit} gives concentration, while Proposition~\ref{prop:invariance} gives $n/\Gamma_n\to1$. For example, $s_n=\lfloor n^\alpha\rfloor$, $0<\alpha<1$, yields
\[
 \E\tau_n\sim\frac{\alpha}{c}n^{2-\alpha}\log n.
\]
Thus $\sqrt n$ spare phases give an $n^{3/2}\log n/(2c)$ mean acquisition scale without a nonzero limiting utilization loss. Finite networks still pay the utilization cost in \eqref{eq:cf-metrics}.

\subsection{What sufficiently large finite persistence preserves}
\label{sec:finite-coupling}

Couple the finite- and infinite-persistence protocols from the same state with all nodes contending using the same access choices. For $\kappa=\kappa_n$, put $q_n=1-(1-p_n)^{n-1}$. For every integer horizon $h\ge1$,
\begin{align}
 \Pp\{\tau_n^{(\kappa_n)}\ne\tau_n\}
 &\le \Pp\{\tau_n>h\}\notag\\
 &\quad+n(1+h/\Gamma_n)q_n^{\kappa_n}.
 \label{eq:finite-coupling}
\end{align}
Indeed, a scheduled singleton's conditional collision probability is at most $q_n$. Iterated conditioning bounds any $\kappa_n$ consecutive failures by $q_n^{\kappa_n}$; each node starts at most $1+h/\Gamma_n$ candidate runs before $h$. Until the first timeout the coupled reservation states coincide, and if acquisition precedes every timeout they coincide thereafter.

Under either main theorem, $\tau_n/n^2$ is tight and $q_n\to1-e^{-c}<1$. For fixed $a>0$, take $h=\lceil an^2\rceil$; the timeout term is $O(n^2q_n^{\kappa_n})$. It therefore vanishes for the concrete sufficient choice
\begin{equation}
 \kappa_n=\lceil d\log n\rceil,\qquad d>\frac{2}{-\log(1-e^{-c})}.
 \label{eq:log-persistence}
\end{equation}
Taking $n\to\infty$ and then $a\to\infty$ makes the acquisition times equal with probability tending to one. This sufficient threshold is not asserted to be necessary or sharp. The distributional and in-probability limits transfer. The finite-$n$ kernel remains specific to permanent reservations; finite-persistence moment convergence requires additional tail control. Arbitrary initial scheduled collisions can incur a recovery delay of order $\kappa_n\Gamma_n$.

\section{Numerical Evaluation}
\label{sec:numerics}

Trials are independent within each configuration. Error bars and cumulative distribution function (CDF) bands are pointwise 95\% Monte Carlo intervals; some are too small to be visible behind the markers or curves.

\subsection{Mean scaling, fluctuations, and an exact check}

Fig.~\ref{fig:scaling} validates the fixed-$c$ laws at $p_n=2.5/n$, with populations up to $n=2^{16}$ and extra sublinear-slack points through $n=2^{18}$ in panel (a). All nodes transmit and collide in slot zero, then begin random access without reservations. This forced startup adds one slot to acquisition and leaves all asymptotic limits unchanged. Panels (a)--(b) confirm mean scaling, using \eqref{eq:Bn} as a finite-sum reference also at fixed slack.

\begin{figure*}[t]
\centering
\includegraphics[width=0.96\textwidth]{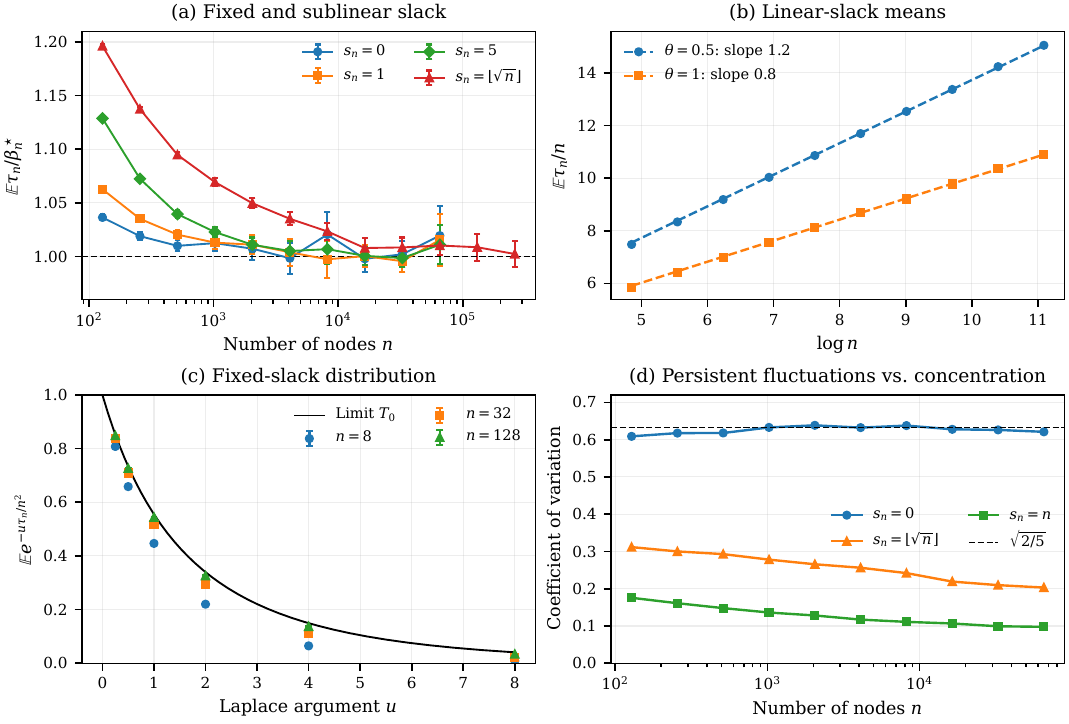}
\caption{RP-TSA acquisition at $np_n=2.5$ under forced startup. (a) Mean acquisition divided by \eqref{eq:Bn}. (b) Linear-slack means; dashed lines have theoretical slopes $(1+\theta)/(c\theta)$ from \eqref{eq:linear}, with vertical offsets fitted to the six largest populations. (c) Empirical Laplace transforms of $\tau_n/n^2$ at $s_n=0$, compared with \eqref{eq:s0-moments}; $n=8,32,128$ with 30000 trials per size. (d) Coefficient of variation: fixed slack retains fluctuations, whereas diverging slack concentrates. Bars in (a)--(c) are pointwise 95\% Monte Carlo intervals.}
\label{fig:scaling}
\end{figure*}

Panels (c)--(d) show distributional convergence and the contrast between persistent fluctuations and concentration. An independent slot-by-slot mean estimate at $n=\Gamma_n=64$, $p_n=2.5/64$, agrees with the exact bracket \eqref{eq:frame-mean}--\eqref{eq:finite-mean-bracket}, with one slot added to both bounds to include the forced initial collision.

\subsection{Selecting the contention intensity}

Fig.~\ref{fig:tuning}(a)--(b) plots normalized RP-TSA acquisition means under ordinary startup over $c\in\{0.5,0.75,\ldots,8\}$. At the comparison sizes $n=32,64$, the observed grid minima select $c=3.75,4.5$, respectively, for $\Gamma_n=n$, and $c=2.5$ for $\Gamma_n=2n$. These values are selected for RP-TSA mean acquisition, and the subsequent comparisons use independent trials.

\begin{figure*}[t]
\centering
\includegraphics[width=0.96\textwidth]{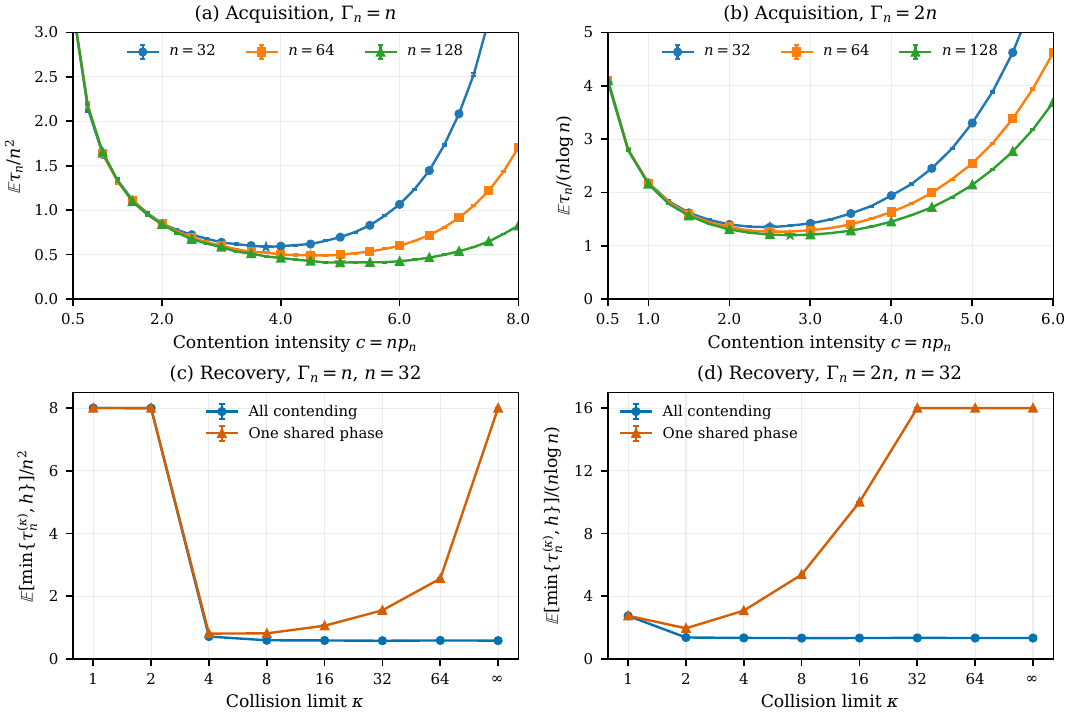}
\caption{Contention and persistence tradeoffs. (a)--(b) RP-TSA mean acquisition from all nodes contending versus $c=np_n$ (6000 trials per point); stars mark grid minima. Selection uses the full sweep to $c=8$; (a)--(b) show normalized means up to $3$ and $5$, respectively, with $c\le6$ in (b). (c)--(d) Restricted means $\E[\min\{\tau_n^{(\kappa)},h\}]$ for $n=32$, with 4000 trials per point, from all nodes contending or all nodes reserving the same phase. The horizons are $h=8192$ and $1774$ slots, using the selected $c$ in each case. The $\kappa=\infty$ endpoint denotes permanent reservations. Bars are pointwise 95\% Monte Carlo intervals.}
\label{fig:tuning}
\end{figure*}

The leading coefficient decreases as $1/c$ for each fixed $c$, but this asymptotic statement is not uniform in $c$. The factor $(1-p_n)^{-(m-1)}$ in \eqref{eq:W-mean} captures the early collision penalty that eventually increases finite-$n$ acquisition time.

\subsection{Baselines, startup, and delivered service}

We compare 1-pTSA, RP-TSA, L-MAC, and L-ZC \cite{Fang2013}. The access protocols use the same $c$, selected for RP-TSA acquisition; 1-pTSA is not independently tuned. An RP-TSA control uses $p_n=1/\Gamma_n$. Each learner chooses one attempt per frame of $\Gamma_n$ slots. L-MAC starts with uniform phase probabilities $w_j$, $j\in\{0,\ldots,\Gamma_n-1\}$. Success assigns probability one to that phase; after failure in phase $J$, the update is $w_j\leftarrow\beta w_j+(1-\beta)\1{j\ne J}/(\Gamma_n-1)$, where $\1{\cdot}$ is the indicator function, and the next phase is drawn accordingly. We use $\beta=0.95$ and the persistent endpoint $\beta=1$. The latter specializes the published rule but lies outside its $\beta\in(0,1)$ convergence theorem. L-ZC retains a successful phase; after collision it retains that phase with probability $1/(\Gamma_n-n+2)$ and otherwise chooses uniformly among the preceding frame's idle phases.

We compare two startup conditions. Under ordinary startup, all 1-pTSA and RP-TSA nodes begin contending and attempt independently with probability $p_n$ per slot; L-MAC and L-ZC nodes independently draw uniform initial phases. As a synchronized stress test, forced startup makes all nodes transmit and collide in slot zero, retaining uniform L-MAC weights. Access nodes resume contention in slot one; learners wait until the next frame.

We use a common frame origin and update learner choices at frame boundaries; the published mechanisms need no shared phase label. Learner acquisition occurs when newly selected phases first become distinct, without a confirmation frame. Goodput uses the same slot clock and horizons for all schemes. The rate control matches nominal attempts per frame.

\begin{figure*}[t]
\centering
\includegraphics[width=0.96\textwidth]{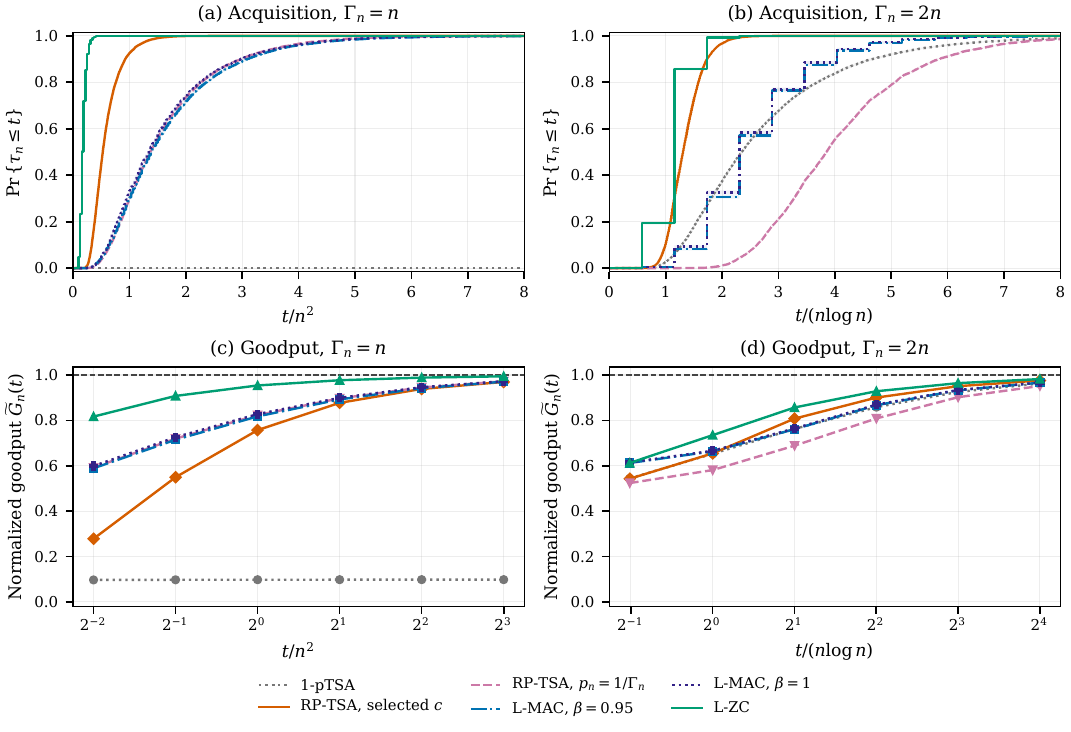}
\caption{Ordinary-start acquisition CDFs and normalized goodput for $n=32$, with 12000 trials per protocol. Both access protocols use $c$ selected for RP-TSA mean acquisition: $3.75$ at $\Gamma_n=n$ and $2.5$ at $\Gamma_n=2n$. The matched RP-TSA control uses $p_n=1/\Gamma_n$. The dashed horizontal reference in (c)--(d) is the collision-free utilization $n/\Gamma_n$, normalized to one. CDF bands are pointwise 95\% Wilson intervals; goodput bars are pointwise 95\% normal Monte Carlo intervals.}
\label{fig:finite-horizon}
\end{figure*}

Fig.~\ref{fig:finite-horizon} uses ordinary startup; Table~\ref{tab:endpoint} compares final-horizon goodput with the synchronized stress test. L-ZC has the smallest observed mean acquisition time and highest goodput in both ordinary-start cases, using additional idle-phase sensing. Selected-intensity RP-TSA acquires faster than persistent L-MAC and matched RP-TSA, but at zero slack delivers slightly lower goodput. Faster acquisition therefore need not yield more packets over a finite horizon.

\begin{table}[t]
\centering
\caption{Normalized goodput for $n=32$ at $8192$ slots ($\Gamma_n=n$) and $1774$ slots ($\Gamma_n=2n$), with 12000 trials per entry. All pointwise 95\% Monte Carlo half-widths are below $0.0005$.}
\label{tab:endpoint}
\small
\renewcommand{\arraystretch}{1.2}
\setlength{\tabcolsep}{4pt}
\begin{tabular}{@{}lcccc@{}}
\toprule
& \multicolumn{2}{c}{$\Gamma_n=n$} & \multicolumn{2}{c}{$\Gamma_n=2n$} \\
Protocol & Ordinary & Forced & Ordinary & Forced \\
\midrule
1-pTSA & 0.0980 & 0.0979 & 0.9626 & 0.9624 \\
RP-TSA, selected $c$ & 0.9693 & 0.9691 & 0.9757 & 0.9750 \\
RP-TSA, $p_n=1/\Gamma_n$ & 0.9722 & 0.9720 & 0.9516 & 0.9510 \\
L-MAC, $\beta=0.95$ & 0.9706 & 0.9667 & 0.9663 & 0.9304 \\
L-MAC, $\beta=1$ & 0.9728 & 0.9689 & 0.9670 & 0.9308 \\
L-ZC & 0.9942 & 0.9888 & 0.9821 & 0.9460 \\
\bottomrule
\end{tabular}
\end{table}

At $\Gamma_n=2n$, persistent L-MAC delivers higher goodput than matched RP-TSA under ordinary startup; a forced collision reverses this ordering. The stress test includes a first-frame delivery disadvantage for learners, so these comparisons depend on initialization.

Table~\ref{tab:acquisition} compares selected and matched RP-TSA with persistent L-MAC. The matched protocols have similar means at zero slack; with linear slack, persistent L-MAC acquires faster. Selected-intensity RP-TSA is faster than both. All table trials acquired before the simulation horizon. Recovery panels instead report acquisition times truncated at horizon $h$.

\begin{table}[t]
\centering
\caption{Sample mean acquisition time (slots), ordinary startup, with 12000 trials at $n=32$ and 6000 at $n=64$. All pointwise 95\% Monte Carlo half-widths are below $1.7\%$ of the sample mean. Selected $c$ uses the RP-TSA grid choice in Fig.~\ref{fig:tuning}.}
\label{tab:acquisition}
\small
\renewcommand{\arraystretch}{1.3}
\setlength{\tabcolsep}{4pt}
\begin{tabular}{@{}rrrrr@{}}
\toprule
$n$ & $\Gamma_n$ & \multicolumn{2}{c}{RP-TSA} & \multicolumn{1}{c}{L-MAC} \\
\cmidrule(lr){3-4}
 & & Selected $c$ & $p_n=1/\Gamma_n$ & $\beta=1$ \\
\midrule
32 & 32 & 605 & 1676 & 1652 \\
32 & 64 & 149 & 453 & 285 \\
64 & 64 & 2021 & 6767 & 6627 \\
64 & 128 & 337 & 1097 & 705 \\
\bottomrule
\end{tabular}
\end{table}

\subsection{Finite persistence: acquisition versus recovery}

Fig.~\ref{fig:tuning}(c)--(d) compares two initial states: all nodes contending, and all nodes reserving phase zero with zero failure counts. In the latter case, finite persistence first produces a deterministic release delay and then acquisition from the former state, denoted $\tau_{\mathrm{cont}}^{(\kappa)}$:
\[
 \tau_{\mathrm{shared}}^{(\kappa)}\stackrel{\mathrm d}{=}d_\kappa+\tau_{\mathrm{cont}}^{(\kappa)},
 \qquad d_\kappa=(\kappa-1)\Gamma_n+1.
\]
For the plotted restricted means, only $h-d_\kappa$ slots remain for acquisition when $d_\kappa<h$; otherwise the mean equals the cap $h$. Infinite persistence never recovers from this shared phase. This special initialization illustrates the cost of releasing a common reservation; Proposition~\ref{prop:global} supplies the general arbitrary-state guarantee.

Small collision limits can destroy valid reservations, while large limits delay removal of shared reservations. The preferred finite limit therefore depends on the initialization and available slack. The logarithmic condition \eqref{eq:log-persistence} is a sufficient asymptotic guarantee, not a finite-network tuning rule.

\section{Conclusion}

Permanent reservation protection yields an exact opportunity chain and frame kernel. A deterministic clock comparison gives a nonconcentrating quadratic acquisition law at fixed slack and concentration under diverging slack. Growing sublinear slack makes acquisition subquadratic while utilization tends to one; linear slack gives the $n\log n$ scale at a fixed utilization cost. The experiments separate acquisition from delivered service and show how contention intensity and startup affect both. Finite persistence enables ACK-only recovery from arbitrary reservations, but its transient differs from the permanent model. Population changes, imperfect feedback, and finite-persistence moment asymptotics remain open extensions.

\section*{Acknowledgment}

The author used OpenAI ChatGPT to assist with manuscript preparation, mathematical exploration, and simulation development. The author is responsible for the final content.

\appendices

\section{Proof of Theorem~\ref{thm:fixed-slack}}
\label{app:fixed}

Since $s_n$ is integer-valued, $s_n=s$ eventually. Fix a truncation level $m_\star$, and let $T_n^{(m_\star)}$ be the time from first reaching $m_\star$ contending nodes to absorption. While the contending count remains $m$, its $m+s$ empty locations are fixed and repeat every $\Gamma_n$ slots. Any $\Gamma_n$ consecutive slots contain exactly $m+s$ opportunities. The time to observe $G_{n,m}$ opportunities therefore differs from $\Gamma_nG_{n,m}/(m+s)$ by at most one frame. Summing this discrepancy over $m\le m_\star$ gives
\begin{equation}
 \left|T_n^{(m_\star)}-\Gamma_n\sum_{m=1}^{m_\star}\frac{G_{n,m}}{m+s}\right|
 \le\Gamma_nm_\star,
 \label{eq:finite-cyclic}
\end{equation}
For fixed $m$, $ng_{n,m}\to cm$ and $g_{n,m}G_{n,m}\xrightarrow{\mathrm d}E_m$. Independence and convergence of geometric moments yield
\begin{equation}
 \frac{T_n^{(m_\star)}}{n^2}\xrightarrow{\mathrm d}\sum_{m=1}^{m_\star}\frac{E_m}{cm(m+s)},
 \label{eq:finite-limit}
\end{equation}
with convergence of the first two moments.

To remove the truncation, the access assumption gives $g_{n,m}^{-1}\le b n/m$ uniformly over $m\le n$, for a constant $b$. Define $\mathcal W_{n,>m_\star}=\sum_{m=m_\star+1}^nG_{n,m}/(m+s)$. The one-sided clock bound \eqref{eq:one-sided-clock}, stopped on first reaching $m_\star$, bounds the preceding elapsed time by $\Gamma_n(\mathcal W_{n,>m_\star}+1)$. Moreover,
\begin{align}
 \E\mathcal W_{n,>m_\star}
 &\le b n\sum_{m>m_\star}\frac1{m(m+s)}
 \le\frac{b n}{m_\star}, \label{eq:entrance-first}\\
 \Var(\mathcal W_{n,>m_\star})
 &\le b n^2\sum_{m>m_\star}\frac1{m^2(m+s)^2}
 \le\frac{b n^2}{m_\star^3}. \label{eq:entrance-second}
\end{align}
Thus the $n^{-2}$-scaled entrance remainder is uniformly small in $L^2$ as $m_\star\to\infty$, after $n\to\infty$. Together with finite-truncation moment convergence, this proves \eqref{eq:fixed-limit}. Finally,
\begin{align}
 \E T_s&=\frac1c\sum_{m\ge1}\frac1{m(m+s)}, \label{eq:Ts-mean}\\
 \Var(T_s)&=\frac1{c^2}\sum_{m\ge1}\frac1{m^2(m+s)^2}. \label{eq:Ts-var}
\end{align}
Partial fractions give \eqref{eq:fixed-mean}--\eqref{eq:fixed-var}. Factoring $m(m+s)+u/c=(m+a_s(u))(m+b_s(u))$ and applying Euler's gamma product gives \eqref{eq:fixed-lt}; at $s=0$, the product for $\sinh$ gives \eqref{eq:s0-moments}.

\section{Proof of Theorem~\ref{thm:slack}}
\label{app:slack}

Write $s=s_n$. The exact weighted work in \eqref{eq:weighted-clock} has moments
\begin{align}
 \mu_n:=\E\mathcal W_n
 &=\frac1{p_n}\sum_{m=1}^n\frac{(1-p_n)^{-(m-1)}}{m(m+s)}, \label{eq:W-mean}\\
 \Var(\mathcal W_n)
 &=\sum_{m=1}^n\frac{1-g_{n,m}}{(m+s)^2g_{n,m}^2}. \label{eq:W-var}
\end{align}
The finite-$n$ access factor affects numerical centering but not the leading fixed-$c$ law. Put $\sigma_{n,s}=\sum_{m=1}^n[m(m+s)]^{-1}$. If $np_n\le b$, then for fixed $\varepsilon>0$,
\[
 0\le(1-p_n)^{-(m-1)}-1\le e^{b\varepsilon+o(1)}-1,
 \qquad m\le\varepsilon n.
\]
This factor is uniformly bounded for all $m\le n$, while
\[
 \frac{\sum_{m>\varepsilon n}[m(m+s)]^{-1}}{\sigma_{n,s}}\longrightarrow0.
\]
Indeed, the numerator is $O(n^{-1})$; for $s\le n$, $\sigma_{n,s}\ge h_s/(2s)$, and for $s>n$, $\sigma_{n,s}\ge h_n/(s+n)$. These bounds and $s\to\infty$, $s=O(n)$ imply
\[
 0\le\frac{p_n\mu_n-\sigma_{n,s}}{\sigma_{n,s}}
 \le e^{b\varepsilon+o(1)}-1+o(1).
\]
Taking $n\to\infty$ and then $\varepsilon\downarrow0$ yields
\begin{equation}
 \mu_n\sim\frac{\sigma_{n,s}}{p_n}=\frac{\beta_n^\star}{\Gamma_n}.
 \label{eq:W-asymptotic}
\end{equation}
Also $g_{n,m}^{-1}\le b/(mp_n)$, so
\begin{equation}
 \Var(\mathcal W_n)\le\frac{b}{p_n^2}
 \sum_{m=1}^n\frac1{m^2(m+s)^2}\le\frac{b}{p_n^2s^2}.
 \label{eq:W-concentration}
\end{equation}
Let $\lambda_{n,s}=h_n+h_s-h_{n+s}$. Since $\mu_n\asymp\lambda_{n,s}/(p_ns)$ and $\lambda_{n,s}\to\infty$,
\[
 \frac{\Var(\mathcal W_n)}{\mu_n^2}=O(\lambda_{n,s}^{-2}),
 \qquad
 \frac{1+n/s}{\mu_n}=O(\lambda_{n,s}^{-1}).
\]
Equation~\eqref{eq:two-sided-clock} now gives
\begin{align*}
 \left\|\frac{\tau_n}{\beta_n^\star}-1\right\|_2
 &\le\frac{\Gamma_n}{\beta_n^\star}\sqrt{\Var(\mathcal W_n)}
 +\left|\frac{\Gamma_n\mu_n}{\beta_n^\star}-1\right|\\
 &\quad+\frac{\Gamma_n(1+n/s)}{\beta_n^\star}=o(1).
\end{align*}
Expanding \eqref{eq:harmonic} gives the sublinear and linear cases.

\small

\end{document}